\documentclass[12pt, reqno]{amsart}
\usepackage[margin=1in]{geometry}
\usepackage{amssymb,amsmath, amsthm,mathtools,tocvsec2,pdflscape,cite}
\usepackage{graphicx} 
\usepackage{pinlabel}
\usepackage[utf8]{inputenc}
\usepackage{amssymb}
\usepackage{caption}
\usepackage{subcaption}
\usepackage{blindtext}
\usepackage{enumitem}
\usepackage{afterpage}
\usepackage{yfonts}
\usepackage{hyperref}
\usepackage{overpic}
\usepackage{amsmath}
\usepackage{mathtools}
\usepackage{latexsym}
\usepackage[normalem]{ulem}
\usepackage{xcolor}
\usepackage[utf8]{inputenc}
\usepackage[T1]{fontenc}
\usepackage[mathscr]{euscript}
\hypersetup{colorlinks=false,linkbordercolor=white,citebordercolor=white}

\numberwithin{equation}{section}
\newtheorem*{theorem*}{Theorem}

\newtheorem{theorem}{Theorem}[section]

\newtheorem{lemma}[theorem]{Lemma}

\newtheorem{corollary}[theorem]{Corollary}
\theoremstyle{definition}

\newtheorem*{remark}{Remark}

\newcommand{\beas}{\begin{eqnarray*}}
\newcommand{\eeas}{\end{eqnarray*}}
\newcommand{\bes} {\begin{equation*}}
\newcommand{\ees} {\end{equation*}}
\newcommand{\be} {\begin{equation}}
\newcommand{\ee} {\end{equation}}
\newcommand{\bea} {\begin{eqnarray}}
\newcommand{\eea} {\end{eqnarray}}

\newcommand{\de}{\delta}
\newcommand{\eps}{\varepsilon}

\newcommand{\Rd}{\mathbb R^d}
\newcommand{\R}{\mathbb R}
\newcommand{\vol}{\text{vol}}
\newcommand{\inte}{\text{int}\,}
\begin{document}
\title[]{A (Hilbert) geometric algorithm for approximating the halfspace depth of a point in a convex body}
\author[]{Purvi Gupta}
\address{Department of Mathematics, Indian Institute of Science, Bangalore}
\email{purvigupta@iisc.ac.in}
\author[]{Anant Narayanan}
\address{Department of Mathematics, Indian Institute of Science, Bangalore}
\email{mkanant17@gmail.com}

\begin{abstract}
Halfspace (or Tukey) depth is a fundamental and robust measure of centrality of data points in multivariate datasets. Computing the depth of a point with respect to the uniform distribution on an open convex body in $\Rd$ is a natural algorithmic problem. While the coarser task of testing membership in convex bodies has been extensively studied, the refined problem of evaluating depth has received comparatively little attention in the literature. 

In this work, we present an algorithm for approximating the halfspace depth of a point in an open convex body $K \subset \Rd.$ To the best of our knowledge, this is the first deterministic algorithm for this problem. As part of our approach, we design an algorithm for answering approximate membership queries for the depth-trimmed regions of $K$ (i.e., the superlevel sets of the depth function). Our data structure is inspired by recent work of Abdelkader and Mount [SOSA 2024], wherein approximate membership queries for $K$ are answered using geometric structures derived from the Hilbert metric on $K.$ A key component underlying our data structure is a novel quantitative comparison between the depth-trimmed regions and the Hilbert metric balls of $K$.

Lastly, to highlight the computational expense of the problem, we present an algorithm for determining the exact depth of a point in an open planar convex polygon presented as the intersection of finitely many halfplanes.
\end{abstract}

\maketitle
\section{Introduction}
\subsection{Testing membership in convex bodies} Given a convex polytope $K \subset \mathbb{R}^d$, the polytope membership query problem asks for an efficient algorithm to determine whether a given query point lies in $K$. While there exist linear-size data structures that answer such queries in logarithmic time when the dimension is at most $3$, in higher dimensions, the best known near-linear-size data structures have a query time of roughly $O(n^{1 - 1/\lfloor d/2 \rfloor}),$ where $n$ is the number of facets of $K$; see \cite{Ma92}. Since this is too costly for practical purposes, one typically considers approximate membership queries instead.

 Let $\eps>0$ be a fixed parameter. Given a convex body $K\subset\R^d$, let $K\oplus \eps=\{x\in\R^d:\text{dist}(x,K)<\eps\}$ and $K\ominus\eps=\{x\in K:\text{dist}(x,bK)>\eps\}$. Given a query point $q\in\R^d$, the $\eps$-(outer) approximate polytope membership query ($\eps$-APM) returns ``Yes'' if $q\in K$, and ``No'' if $q\notin K\oplus\eps\cdot\text{diam}(K)$. On the other hand, the $\eps$-(inner) approximate polytope membership query ($\eps$-APM) returns ``Yes'' if $q\in K\ominus\eps\cdot\text{diam}(K)$, and  ``No'' if $q\notin K$. We refer to either of these as an $\eps$-APM query. There are many efficient algorithms in the literature for answering APM queries; see \cite{Du74, BFP82, AFM18} and references therein for some of the earlier results. These typically rely on classical partitioning schemes such as grids and quadtrees.

 An optimal algorithm for APM queries with storage space $O\left(1/\eps^{(d-1)/2}\right)$ and query time $O(\ln(1/\eps))$ was first given by Arya, da Fonseca, and Mount in \cite{AFM17b}. Subsequently, two more algorithms achieving the same bounds were proposed by Abdelkader and Mount in \cite{AbMo18} and \cite{AbMo24}. The optimality of these three algorithms, compared to the earlier ones, is largely due to the use of shape-sensitive geometric structures such as Macbeath regions and the Hilbert metric.  

In this paper, we study a refined version of the membership problem: given an open convex body $K\subset\R^d$, and a query point $q\in\R^d$, how ``deep'' is $q$ within $K$? We describe a ray-shooting approach, similar to that of Abdelkader and Mount in \cite{AbMo24}, to (approximately) answer such a query. Our analysis uses the Hilbert metric in a crucial way. We view this work as part of a growing body of literature on the application of the Hilbert metric to algorithm design; see \cite{GeMo23, AbMo24, BDGMSYM, NiSu19, NiSu23}.

\subsection{The halfspace depth of a point in a convex body} In statistics, the halfspace (or Tukey) depth of a point is a measure of how deeply it lies within a fixed point cloud in $\Rd.$ The deepest points serve as higher-dimensional analogues of the one-dimensional median. See \cite{NSW19} for a discussion on the robustness of halfspace depth as a notion of centrality. The computational problem of designing an efficient algorithm to compute the (exact or approximate) halfspace depth of a point with respect to a given data set of
$n$ points has been extensively studied; see \cite{RuRo96, MRRSS03, Ch04, Dy16, DyMo16, LMM19}, for instance.

More generally, the halfspace depth of a point $q \in \mathbb{R}^d$ with respect to a probability distribution $\mu$ on $\mathbb{R}^d$ is defined as
\[
D_\mu(q) = \inf\left\{\mathbb{P}(\langle v, X - q \rangle \geq 0) : v \in \mathbb{R}^d,\, \|v\| = 1 \right\},
\]
where $X$ is a random variable with distribution $\mu$, and $\langle \cdot, \cdot \rangle$ denotes the standard inner product on $\mathbb{R}^d$. In particular, if $\mu = \mu_K$ is the uniform probability measure supported on an open convex body $K \subset \mathbb{R}^d$, then
\[
D_K(q) = D_{\mu_K}(q) = \inf\left\{ \frac{|\{ z \in K : \langle v, z - q \rangle \geq 0 \}|}{|K|} : v \in \mathbb{R}^d,\, \|v\| = 1 \right\},
\]
which is the minimum possible normalized volume of a cap of $K$ cut off by a hyperplane passing through $q.$

A natural computational problem is to design an algorithm that computes the halfspace depth of a query point $q \in \Rd$ with respect to a given open convex body $K \subset \Rd$. We refer to this as the \emph{exact depth query} for $K$. Observe that all points in $K$ have positive depth, while points outside $K$ have depth zero. Hence, the exact depth query for $K$ is a refinement of the exact membership query for $K$, and thus, at least as computationally expensive.

As far as we are aware, exact depth queries for convex polytopes have not been addressed in the literature. In Section~\ref{S:EDQ}, we show that any open planar convex polygon $K$, given as the intersection of $n$ halfplanes, can be preprocessed in $O(n \log n)$ time into a data structure of size $O(n)$ that answers the exact depth query for $K$ in $O(n^3)$ time for any query point $q\in \R^2$. In higher dimensions, this approach yields significantly worse bounds, making it prudent to allow for errors.

\subsection{The approximate depth query problem} The focus of this paper is the following problem. Given an open convex body $K \subset \mathbb{R}^d,$ and $\eps \in (0, 1/3)$, design a data structure and an algorithm that, for any query point $q \in \mathbb{R}^d$, returns a value $\mathfrak{D}(q)$ satisfying
 \[ \begin{cases}
\mathfrak{D}(q) = \varepsilon, \hspace{2mm}\text{if } {D_K(q) \leq \varepsilon}, \\
(1 - \varepsilon)\, D_K(q) \leq \mathfrak{D}(q) \leq (1 - \varepsilon)^{-1}D_K(q), \hspace{2mm}\text{otherwise.}
\end{cases}
\]
We refer to this as the \emph{$\eps$-approximate depth query} ($\eps$-ADQ) for $K.$ Note that we do not impose a multiplicative guarantee for all values of $D_K(q)$, since this would require solving the exact membership query for $K$ when $D_K(q) = 0$, which is computationally expensive. Building on recent geometric developments, we establish the following result.

\begin{theorem}\label{T:ADQ}Given an open convex body $K\subset\R^d$, and $\eps\in(0, 1/3),$ there is a data structure that answers the $\eps$-ADQ for $K$ with
    \[\text{Storage space}: O\left(\frac{1}{\eps^{\frac{5d+1}{2}}}\right)\hspace{5mm}\text{Query time}: O\left(\frac{1}{\eps^2}\left(\ln\frac{1}{\eps}\right)^2\right).\]
The constant factors in the storage space and query time depend only on $d$ (not on $K$ or $\eps$).
\end{theorem}
To the best of our knowledge, this is the first deterministic algorithm in the literature that answers approximate depth queries for open convex bodies. Our focus is on the existence of this data structure; we do not address preprocessing in this paper. Our storage space and query time bounds hold irrespective of $K$’s representation. The sharpness of our bounds remains an open question. Our approach involves testing \emph{approximate} membership in suitable superlevel sets of the depth function of $K$.

\subsection{Testing membership in depth-trimmed regions}\label{SSS:AMQ} Given an open convex body $K \subset \mathbb{R}^d$ and a parameter $\de \in (0, 1/2],$ the $\de$-depth-trimmed region of $K$ is defined as  
\be\label{E:DTR}
K_\de = \left\{ q \in \Rd : D_K(q) \geq \de \right\}.
\ee
These regions are also referred to as convex floating bodies in the literature. Specifically, the $\delta$-depth-trimmed region of $K$ coincides with the $\de\,|K|$-convex floating body of $K$. Convex floating bodies were independently introduced by Bárány and Larman~\cite{BaLa88} and by Schütt and Werner~\cite{ScWe90}, and have since attracted significant attention in both theory and applications; see~\cite{NSW19} and~\cite{We22}, for instance.

The approximate depth query problem is addressed by directly relating it to the following approximate membership problem. Let $\delta\in (0, 1/2]$ be sufficiently small such that $K_\de\neq\emptyset$, and let $\eps\in(0,1)$. Given a query point $q\in\Rd$, define the {\em $\eps$-approximate membership query for the $\de$-depth-trimmed region} ($\eps$-AMQ for the $\delta$-DTR) of $K$ as a query that returns
    \beas
        \begin{cases}
        \text{Yes, if }q\in K_\de,\\
        \text{No, if }q\notin K_{(1-\eps)\,\de},
        \end{cases}
    \eeas
and either, otherwise.

Motivated by the work of Abdelkader and Mount~\cite{AbMo24}, we design a data structure and an algorithm for answering  the $\eps$-AMQ for the $\de$-DTR of $K$ with storage space $O\left(\dfrac{1}{\eps^{d}\,\de^{\frac{3d-1}{2}}} \right)$ and query time $O\left(\dfrac{1}{\eps\,\de}\ln{\dfrac{1}{\de}}\right)$; see Theorem~\ref{T:AMQ}. Combined with a binary search procedure, this yields our main result, stated as Theorem~\ref{T:ADQ}. 

To our knowledge, the only related work is due to Anderson and Rademacher \cite{AnRa20}. For any fixed centrally symmetric open convex body $K \subset \mathbb{R}^d$, they present an algorithm which, given a sample of $N$ i.i.d.\ points from $K$, answers an $\eps$-weak membership query for $K_\delta$ with probability at least $1 - \tau$ over the sample, using time and sample complexity $\mathrm{poly}(d, 1/\delta, 1/\eps, \ln(1/\tau))$. For a query point $q \in \mathbb{R}^d$, the algorithm must either assert that $q \in K_\delta \oplus \eps$ or that $q \notin K_\delta \ominus \eps$.

\subsection{Our techniques and the role of the Hilbert metric} The bulk of the effort lies in the design and analysis of a data structure and an algorithm for answering the $\eps$-AMQ for the $\de$-DTR of $K$. The data structure primarily consists of the intersection graph $G$ of a collection of \emph{special} ellipsoids contained in $K$. The centers of these ellipsoids form a Delone set in $K_\delta$, which helps control both the number of vertices of $G$ and their maximum degree, and thereby the storage space of the data structure. The algorithm adopts a ray-shooting approach, where the query time is effectively determined by the number of suitably chosen ellipsoids traversed by a ray shot from a prescribed root towards the query point. This is motivated by a similar approach by Abdelkader and Mount~\cite{AbMo24} for answering APM queries.

The analysis of the algorithm relies on properties of the Hilbert geometry of $K$. The Hilbert metric is a complete, projectively invariant metric on $K$, with the key feature that straight lines (in the Euclidean sense) serve as shortest paths. It is known that in the Hilbert geometry of $K$, the Busemann volume (i.e., the $d$-dimensional Hausdorff measure induced by the Hilbert metric) of metric balls grows at most exponentially and at least polynomially (of degree $d$) with the radius. This geometry naturally arises in our context because the aforementioned \emph{special} ellipsoids closely approximate Hilbert metric balls. Since the centers of these ellipsoids form a Delone set in $K_\de$, they give rise to a disjoint collection of Hilbert metric balls of fixed radius, allowing us to bound the number of vertices of $G$---and hence, the storage space of the data structure---via an upper bound on the Busemann volume of $K_\de$.
The Delone set property also yields a lower bound on the Hilbert distance covered by the (query) ray as it passes through each suitably chosen ellipsoid on its path from the root to the query point. Combined with an upper bound on the Hilbert diameter of $K_\delta$, this gives an upper bound on the number of ellipsoids traversed---and hence, the query time. Upper bounds on the Busemann volume and Hilbert diameter of $K_\delta$ follow from the final component of our analysis: a novel quantitative comparison between the depth-trimmed regions and the Hilbert metric balls of $K$; see Theorem~\ref{T:FB-HB} and Corollary~\ref{C:bound}. These results may be of independent theoretical interest. 

\subsection{Organization of the paper} Some mathematical preliminaries are collected in Section~\ref{SS:prelim}. While most results therein are standard, Theorem~\ref{T:FB-HB} and Corollary~\ref{C:bound} are new contributions. In Section~\ref{S:ADQ}, we show that, modulo an algorithm for answering the $\eps$-AMQ for the $\de$-DTR of $K$, a binary search procedure can be used to answer the $\eps$-ADQ for $K,$ thus proving Theorem~\ref{T:ADQ}. An algorithm for the $\eps$-AMQ is presented in Section~\ref{SS:DOA}. In Section~\ref{S:EDQ}, we describe an algorithm for answering exact depth queries for any open planar convex polygon given as the intersection of $n$ halfplanes. Some proofs and remarks from Section~\ref{SS:prelim} are deferred to the appendix.

\subsection{Acknowledgements} 
We are grateful to Rahul Saladi for sharing his insights on the subject, which vastly improved our understanding of the existing literature. We are also very grateful to David M. Mount for patiently answering all our questions regarding the data structure constructed in \cite{AbMo24}.

P. Gupta is supported by the SERB grants MATRICS (MTR/2023/000393) and WERG (WEA/2023/000017). She would also like to acknowledge support from the ICTP through the Associates Programme (2023-2028). A. Narayanan is supported by a scholarship from the Indian Institute of Science. Both the authors are supported by the DST-FIST programme (grant no. DST FIST-2021 [TPN-700661]). 

\section{Mathematical preparation}\label{SS:prelim}
We begin this section by introducing some notation, followed by a collection of preliminary results and observations. We work in the $d$-dimensional Euclidean space $\mathbb{R}^d$, equipped with the standard Euclidean norm $\|\cdot\|$. Throughout this paper, $o$ denotes the origin in $\mathbb{R}^d$, and $\mathbb{B}(x, r)$ denotes the open Euclidean ball of radius $r > 0$ centered at $x \in \mathbb{R}^d$.  The interior and boundary of a set $S \subset \Rd$ are denoted by $\inte S$ and $bS$, respectively. The Lebesgue volume of a measurable set $S \subset \mathbb{R}^d$ is denoted by $|S|$; in particular, $\omega_d$ denotes $|\mathbb{B}(o,1)|$. We write $\#A$ to denote the cardinality of a finite set $A.$

We use standard asymptotic notation throughout to absorb constant factors (depending only on the dimension 
$d$, unless stated otherwise). For instance, we say that a real-valued function $f$ of $x\in A\subset\R$ is $O(x)$ if $f(x)\leq cx$ for all $x\in A,$ where $c>0$ is a constant. We use $\log_{\,b}$ to denote the logarithm to base $b$, and write ``$\ln$'' for the natural logarithm (base $e$).  
\begin{figure}[h]
\begin{overpic}[grid=false,tics=10,scale=0.6]{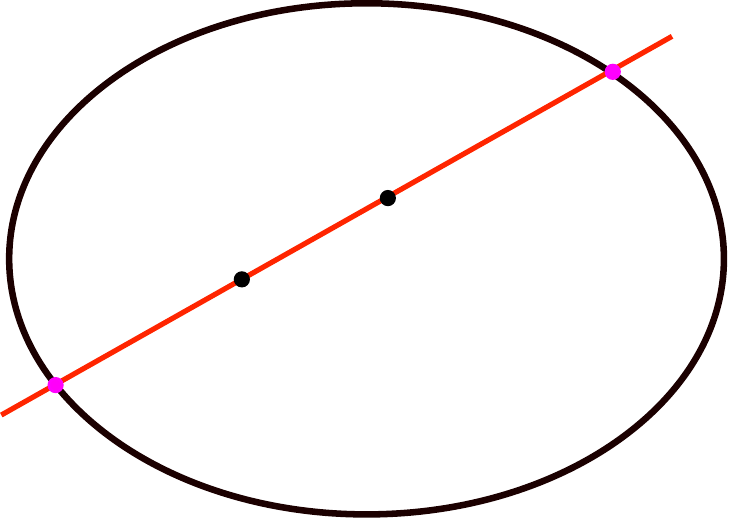}
\put(35,28){$p$}
    \put(55,40){$q$}
    \put(89,59){$b$}
     \put(5,12){$a$}
\end{overpic}
    \caption{Hilbert metric on $K.$}
    \label{fig:mesh1}
\end{figure}

Let $K\subset\Rd$ be an open convex body, i.e., a nonempty, bounded, and open convex set. For $p,q\in K$, let $a,b\in bK$ denote the points where the line joining $p$ and $q$ meets $bK$ so that $a,p,q,$ and $b$ are in consecutive order. Then, the {\em Hilbert metric} on $K$ is given by 
\bes
    d_K(p,q)=\frac{1}{2}\ln \frac{\Vert q-a\Vert \Vert p-b\Vert}{\Vert p-a\Vert\Vert q-b\Vert}.
\ees
We introduce some additional notation and definitions.
\begin{itemize}
\item $\vol_K$ denotes the Busemann volume, i.e., the $d$-dimensional Hausdorff measure associated with the metric space $(K, d_K).$
    \item  $\lambda K = \{ \lambda y : y \in K \}$ if $\lambda > 0$, and $\lambda K = \emptyset$ if $\lambda \leq 0$; defined if $o\in K.$
    \item $K$ is said to be in canonical form if its inner John ellipsoid (i.e., the open ellipsoid of maximal volume contained in $K$) is $\mathbb{B}(o, 1/(2d)).$
  \end{itemize}
\noindent For $x\in K$ and $\lambda\geq0,$
  \begin{itemize}
   \item $B_K(x, r) = \{ y \in K : d_K(x, y) < r \}$ if $r > 0$, and $B_K(x, r) = \emptyset$ if $r \leq 0.$
      \item $\text{ray}(x) =\| x-p\|$, where $p\in bK$ denotes the point where the ray from $o$ to $x$ meets $bK;$ defined if $o\in K.$
    \item ${\delta(x)}$ denotes the Euclidean distance of $x$ from $bK.$ 
   \item 
   $M^\lambda(x)=x+\lambda\left((K-x)\cap (x-K)\right)$ is the $\lambda$-Macbeath region at $x.$
       \item $E^\lambda(x)$ denotes the $\lambda$-Macbeath ellipsoid at $x$, i.e., the inner John ellipsoid of $M^\lambda(x).$ 
   
\end{itemize}
With the notation and definitions in place, we begin by noting some useful properties of the depth-trimmed regions of $K$ as defined in~\eqref{E:DTR}.

\begin{lemma}\label{L:FB} Let $K\subset\Rd$ be an open convex body, and let $\de\in(0, 1/2].$ Then  the following statements hold.
\begin{itemize}
    \item [$(a)$] For any affine transformation $A:\Rd\rightarrow\Rd$, $A(K_\de)=(AK)_{\de}$.
    \item [$(b)$] Through every point of $bK_\de$, there is at least one supporting hyperplane $H$ of $K_\de$ that cuts off a set of volume ${\de\,|K|}$ from $K$. Moreover, the hyperplane $H$ touches $K_\de$ at exactly one point, which is the centroid of $H\cap K$.
    \item [$(c)$] Each nonempty $K_\de$ is a strictly convex compact set.
\end{itemize}
\end{lemma}
\noindent Statement $(a)$ follows directly from the definition. See \cite[Lemma 2]{ScWe94} for proofs of $(b)$ and $(c)$.

In the rest of this section, we assume that the open convex body $K\subset \Rd$ satisfies  
\be\label{E:can_form_a}
\mathbb B(o,\gamma)\subset K\subset \mathbb B(o,1/2),\ee 
where $\gamma\in (0, 1/2]$ is a fixed constant.

We now turn to some properties of Hilbert metric balls that will play a crucial role later in the paper. The first result, stated in \cite[Lemma $11$]{VeWa21} in terms of the so-called asymptotic balls, can be reformulated as a comparison of Hilbert metric balls with suitable dilates of $K$.

\begin{lemma}[\hspace{1sp}{\cite[Lemma 11]{VeWa21}}]\label{L:DilatesHilb} For any $r\in\R$,
\bes 
B_K(o, r)\subset \left(1-e^{-2r}\right)K
\subset B_K\left(o,r+\frac{1}{2}\ln \left(1+\frac{1}{2\gamma}\right)\right).
\ees   
\end{lemma}

The following result captures the volume growth behavior of Hilbert metric balls, and is a combination of results from \cite[Theorem A]{Th17} and \cite[Theorem $2.1$]{Ve13}.

\begin{lemma}\label{L:VolBound}
For any $x\in K$ and $r>0,$
\bes \normalfont{\vol}_K B_K(x, r)= \Omega(r^d) \quad \quad \text{and} \quad \quad \normalfont{\vol}_K B_K(x, r)= O(e^{(d-1)r}).\ees
\end{lemma}

The next lemma shows that Macbeath ellipsoids act as approximants to Hilbert metric balls.

\begin{lemma}[{\hspace{1sp}\cite[Corollary~2.1]{AbMo24}}]\label{L:HilbMac} Let $x\in K$ and $\lambda\in (0,1)$ be fixed. Then
\bes 
B_K\left(x,\frac{1}{2}\ln\left(1+\frac{\lambda}{\sqrt{d}}\right)\right)\subset E^{\lambda}(x)\subset B_K\left(x, \frac{1}{2}\ln\frac{1+\lambda}{1-\lambda}\right).
\ees   
\end{lemma}
An important expansion-containment property of Macbeath ellipsoids is captured in the following lemma.
\begin{lemma}[\hspace{1sp}{\cite[Corollary~3.1]{AbMo24}}]\label{L:ellips}Let $\lambda\in (0,1)$ be fixed. Suppose $x,y\in K$ are such that $E^{\lambda}(x)\cap E^{\lambda}(y)\neq \emptyset.$ Then, for any $\alpha\geq 0$ and $\beta=\dfrac{2+\alpha(1+\lambda)}{1-\lambda}\sqrt{d},$
\[E^{\alpha\lambda}(y)\subset E^{\beta \lambda}(x).\]
\end{lemma}

 A {\em cap} $C$ of $K$ is a nonempty intersection of $K$ with a (closed) halfspace. Let $H$ denote the boundary hyperplane of this halfspace. The {\em base} of $C$ is the set $K\cap H$.
 The {\em apex} of $C$ is any point on $bC\cap bK$ such that a supporting hyperplane of $K$ at this point is
 parallel to $H$. The {\em width} of $C$ is the distance between $H$ and this supporting hyperplane.  The following result is a combination of results from \cite{VeWa21} and \cite{AFM24}, but we include a proof in the appendix for the sake of completeness. 
 \begin{lemma}\label{L:width} Let $C$ be a cap of $K$ and let $x$ denote the centroid of the base of $C.$ Suppose $o\notin \{x\}\cup \normalfont{\inte} C.$ Then
\be\label{E:ray_vol}
 \frac{\sqrt{3}\gamma}{d}\normalfont{|C|}< \normalfont{\text{ray}}(x)<
 \left(\frac{2}{\omega_d}\right)^{1/d}
 \frac{|C|^{1/d}}{4\gamma^2}.
\ee

\end{lemma}
 The next result is obtained by slightly modifying the proof of \cite[Lemma~4.2]{AFM17}. 

 \begin{lemma}[\hspace{1sp}{\cite[Lemma~4.2]{AFM17}}]\label{L:width2} For any point $x\in K\setminus\{o\},$ 
\[\text{\normalfont ray}(x)\leq \frac{1}{2\gamma}\delta(x).\]    
\end{lemma} 

Finally, a crucial ingredient in the design and analysis of our data structure and algorithm in Section~\ref{SS:DOA} is the following comparison between the depth-trimmed regions of $K$ and metric balls in the Hilbert geometry of $K,$ both of which are affine invariant. Since the proof is technical, we defer it to the appendix, where we also comment on the sharpness of the result. 

\begin{theorem}\label{T:FB-HB} Let $K\subset\mathbb{R}^d$ be an open convex body. Suppose that $\mathbb B(o,\gamma)\subset K\subset \mathbb B(o,1/2)$ for some $\gamma\in(0,1/2]$. Let $\de\in (0, 1/2]$ be such that $o\in K_{\de}.$ Then
\be
B_K\left(o, \frac{1}{2d}\ln\frac{1}{\de}+\frac{1}{2d}\ln \dfrac{2^{2d-1}\gamma^{3d}\omega_d}{|K|}\right)%
\subset K_\de\subset %
 B_K\left(o,\frac{1}{2}\ln\frac{1}{\de}+\frac{1}{2}\ln \frac{d(1+2\gamma)}{4\sqrt{3}\gamma^2\,|K|}\right),\label{E:main_inc}
\ee
where the ball in the left-hand side is nonempty only when ${\de\,|K|}<2^{2d-1}\gamma^{3d}\omega_d.$
\end{theorem}

As a consequence of the above theorem, we obtain a bound on the Hilbert distance between any point in the $\de$-depth-trimmed region and a fixed point within it.

\begin{corollary}\label{C:bound} Let $K\subset\mathbb{R}^d$ be an open convex body in canonical form. Let $\de\in(0,1/2]$ be sufficiently small such that $K_\de\neq \emptyset.$ There exists a point $x_\de\in K_\de$ such that
\be\label{E:main_yeps} 
K_\de\subset B_K\left(x_\de, \frac{1}{2}\ln \frac{1}{\de}+24d^3\right).\ee
Specifically, $x_\de$ can be chosen as $o$ when $o\in K_\de$, and as any point in $K_\de$ when $o\notin K_\de$.
\end{corollary}  
\begin{proof} 
To prove \eqref{E:main_yeps}, first assume that $o\in K_\de$. Since $K$ is in canonical form, John's theorem implies that $\mathbb{B}(o, 1/(2d)) \subset K \subset \mathbb{B}(o, 1/2)$. Thus, by \eqref{E:main_inc} and the fact that $|K|\geq \omega_d/(2d)^d$, we have that
\[K_\de\subset B_K\left( x_{\de},\frac{1}{2}\ln \frac{1}{\de}+\underbrace{\frac{1}{2}\ln \frac{(d+1)(2d)^{d+2}}{4\sqrt{3}\,\omega_d}}_{q_d}\right),
\]
where $x_{\de}=o$. 

Next, suppose $o\notin K_\de\neq \emptyset$. Since any hyperplane through $o$ cuts off a cap of volume at least $\dfrac{\omega_d}{2(2d)^d}$ from $K$, it must be that ${\de\,|K|}\geq  \dfrac{\omega_d}{2(2d)^d}.$ Fix an $x_\de\in K_\de.$ {Consider the ray $\ell$ from $o$ through $x_\de$, and let $x_\delta'$ denote the point on $\ell$ farthest from $o$ that lies in $bK_\delta.$ Note that $\text{ray}(x_{\de})\geq \text{ray}(x_{\de}')$. Moreover, by Lemma~\ref{L:FB}, there is a cap $C$ of $K$, of volume ${\de\,|K|}$, cut off by a supporting hyperplane $H$ of $K_\de$ at $x_\de'$, and $x_{\de}'$ is the centroid of $K\cap H$. Now, there are two possibilities. Either the ray $\ell$ passes through the nonempty interior of $K_\de$, or it is tangential to $K_\de$ and intersects $bK_\de$ at exactly one point (due to the strict convexity of $K_\de$), which is $x_\de'=x_\de$.}

In the former case,  $\inte K_\de$ and $o$ must lie on the same side of the hyperplane $H$. Thus, $o\notin C$, and by Lemma~\ref{L:width} and Lemma~\ref{L:width2}, we obtain that
    \bes 
    \de(x_\de)\geq \frac{1}{d}\text{ray}(x_\de)\geq \frac{1}{d}\text{ray}(x_\de')> \frac{\sqrt{3}}{2d^3}{\de\,|K|}
    \geq \frac{2\sqrt{3}\,\omega_d}{(2d)^{d+3}},
\ees
{where we have used that $\mathbb B(o,1/2d)\subset K\subset \mathbb B(o,1/2)$. In the latter case, let $C'$ denote the cap of $K$ that is complementary to $C$. Then, either $o\in H$,  $o\notin C$, or $o\notin C'$. In all three cases, Lemma~\ref{L:width} applies, and we obtain that}
    \bes 
    \de(x_\de)\geq \frac{1}{d}\text{ray}(x_\de)=\frac{1}{d}\text{ray}(x_\de')> \frac{\sqrt{3}}{2d^3}
    \min\{1-\de,\de\}\,|K|
    \geq \frac{2\sqrt{3}\,\omega_d}{(2d)^{d+3}}.
\ees
It follows that $\mathbb B\left(x_\de, r_d\right)\subset K\subset \mathbb B(x_\de,1),$ where $r_d=\dfrac{2\sqrt{3}\,\omega_d}{(2d)^{d+3}}$. Applying the affine map
    \bes
        T:x\mapsto \frac{1}{2}(x-x_\de),\qquad x\in\Rd,
    \ees
we obtain from Theorem~\ref{T:FB-HB} that 
\[(TK)_\de\subset B_{TK}\left(o,\frac{1}{2}\ln\frac{1}{\de}+\frac{1}{2}\ln \frac{d(1+r_d)}{\sqrt{3}\,r_d^2\,|TK|}\right).
\]
Now, using that $|TK|\geq \omega_d (r_d/2)^d$ and the affine invariance of depth-trimmed regions and Hilbert metric balls, we obtain that 
    \bes
    K_\de\subset B_K\left(x_\de,\frac{1}{2}\ln\frac{1}{\de}+\underbrace{\frac{1}{2}\ln \frac{d2^d(1+r_d)}{\sqrt{3}\,\omega_d\,r_d^{d+2}}}_{q_d'}\right).
    \ees
Then, any choice of $\theta\geq \max\{q_d,q_d'\}$ yields \eqref{E:main_yeps}. We may choose, for instance, $\theta=24d^3$.
\end{proof}

\section{An algorithm for answering the $\boldsymbol{\eps}$-ADQ for $K$}\label{S:ADQ} 

Assuming access to an algorithm for answering exact membership queries for the depth-trimmed regions of $K$, one way to answer the $\eps$-ADQ for a query point $q \in \Rd$ is to perform a binary search to output the largest parameter $\delta$ in the decreasing list
\bes
\left\{ \de_j = \dfrac{1}{2}(1 - \eps)^j :\, j \in \mathbb{N},\, \de_j > \eps \right\} \bigcup\, \{\eps\},
\ees
for which $q \in K_\delta$. If no such $\delta$ exists, then $D_K(q) < \eps$, and we output $\eps$.

We adopt a similar approach in our algorithm, but instead of assuming access to an exact membership oracle, we rely on the following result to answer \emph{approximate} membership queries for the depth-trimmed regions of $K$---more precisely, the $\eps$-AMQ for the $\de$-DTR of $K$, as stated in Subsection~\ref{SSS:AMQ}.


    \begin{theorem}\label{T:AMQ}  Given an open convex body $K\subset \mathbb{R}^d$, $\de\in (0,1/2]$ sufficiently small such that $K_{\de}\neq \emptyset,$ and $\eps\in (0, 1),$ there is a data structure that answers the $\eps$-AMQ for the $\de$-DTR of $K$ with
  \[\text{Storage space}: O\left(\frac{1}{\eps^{d}\,\de^{\frac{3d-1}{2}}} \right) \qquad\text{Query time}: O\left(\frac{1}{\eps\,\de}\ln{\frac{1}{\de}}\right).\]
  The constant factors in the storage space and query time depend only on $d$ (not on $K$, $\eps$, or $\de$).
\end{theorem}

Assuming the above result for the moment, we describe the data structure and algorithm claimed in Theorem~\ref{T:ADQ}. We will prove Theorem~\ref{T:AMQ} in the next section.

Let $l\in\mathbb{Z}_{\geq 0}$ be the largest non-negative integer satisfying 
\bes \frac{1}{2}(1-\eps)^l>\eps.
\ees
Since $\eps<1/3,$ we have that $l\geq 1$ and $l\leq\log_{\,1-\eps}2\eps.$ Let $m=l+1.$ Define, for $1\leq j\leq m,$
\[ \de_j=\begin{cases}
        \dfrac{1}{2}(1-\eps)^j, \hspace{2mm}\text{if} \hspace{1mm} 1\leq j<m,\\
        \eps, \hspace{2mm}\text{if}\hspace{1mm} j=m.
        \end{cases}\]
The data structure stores 
\begin{itemize}
\item [(a)] the ordered list $(1,...,m)$,
        \item [(b)] at each $j$, $1\leq j\leq m$, the information about whether $K_{\de_j}$ is empty or not,
    \item [(c)] at each $j$, $1\leq j\leq m$, if $K_{\de_j}\neq \emptyset$, the data structure (granted by Theorem~\ref{T:AMQ}) that answers the $\eps$-AMQ for the $\delta_j$-DTR of $K$.
\end{itemize}
Using that $l\leq\log_{\,1-\eps}2\eps,$ we have that the storage space of this data structure is
\[O\left(m+\sum_{j=1}^m\frac{1}{\eps^d\,\de_j^{\frac{3d-1}{2}}}\right)=O\left(\frac{1}{\eps^{\frac{5d+1}{2}}}\right).\]

We employ a binary search algorithm and rely on the monotonicity of $K_{\de_j}$ in $j.$ Define the \emph{midpoint} of a finite ordered list of $m$ elements as the element at position \( \lceil m/2 \rceil \). Let $q\in\R^d$ be a query point. Given $a\leq b$ in $\{1,...,m\}$, consider the following procedure.


\begin{itemize}
    \item [(i)] Set $x\leftarrow$ midpoint of the ordered list $(a,...,b)$.
   \item[(ii)] If $x=m$, stop and declare $\mathfrak{D}(q)=\de_m=\eps$.
   \item [(iii)] If $x<m$, $K_{\de_x}$ is empty, and $x=b$, stop and declare $\mathfrak{D}(q)=\de_{x+1}$.
   \item [(iv)] If $x<m$, $K_{\de_x}$ is empty, and $x<b$, set $a\leftarrow x+1,$ {and keep $b$ as is.}
   \item [(v)] If $x<m$, and $K_{\de_x}$ is nonempty, answer the $\eps$-AMQ for the $\de_x$-DTR of $K$ for $q$. 
    \begin{itemize}
        \item  [(a)] If the answer is ``No'' and $x=b$, stop and declare $ \mathfrak{D}(q)=\de_{x+1}$.
         \item  [(b)] If the answer is ``No'' and $x<b$, set $a\leftarrow x+1,$ and keep $b$ as is.
          \item  [(c)] If the answer is ``Yes'' and $x=a$, stop and declare $\mathfrak{D}(q)=\de_x$.
           \item  [(d)] If the answer is ``Yes'' and $x>a$, set $b\leftarrow x-1,$ {and keep $a$ as is.}
    \end{itemize}
    \end{itemize}
Our algorithm consists of running the above {procedure} with $a\leftarrow 1$ and $b\leftarrow m$, and reiterating it with the new values of $a$ and $b$ if the {procedure} does not stop {and declare $\mathfrak{D}(q).$}  It is clear that the algorithm will necessarily {terminate} after a finite number of {iterations}. We demonstrate this with {a specific example.}

Suppose $m=5$. In the first iteration, $a\leftarrow 1$, $b\leftarrow 5$ and $x\leftarrow 3$. Since {$3<5$}, the algorithm checks whether $K_{\delta_3}$ is empty or not. Suppose it is. By (iv), {the algorithm} sets $a\leftarrow 4,$ and keeps $b$ as is. In the second iteration, $x\leftarrow4$. Since {$4< 5$}, the algorithm checks whether $K_{\de_4}$ is empty or not. Suppose it is not. By (v), the algorithm answers the $\eps$-AMQ for the $\de_4$-DTR of $K$ for $q.$ Suppose the response is ``Yes''. Then, since $x=a=4$, by (c), the algorithm {terminates} and declares $\mathfrak{D}(q)=\de_4$. In the final step above, suppose the {response is} ``No''. Then, (b) applies and the algorithm sets $a\leftarrow 5,$ and keeps $b$ as is. In the next iteration, $x\leftarrow 5$. Since $x=m$ in this case, (ii) applies and $\mathfrak{D}(q)$ is declared as $\eps$.

We now check for correctness.
  
\begin{enumerate}
 
 \item[(i)] If $D_K(q) \leq \varepsilon$, then at each iteration, either $K_{\delta_x}$ is empty or the algorithm reports ``No'', continuing until the final stage where $x = m.$ At this point, by (ii), the algorithm terminates and declares $\mathfrak{D}(q)=\eps.$
        \item[(ii)]If $D_K(q)>\eps,$ there exists a smallest index $j\in \left\{1,..., m\right\}$ for which $K_{\de_j}\neq\emptyset$ and the $\eps$-AMQ for the $\delta_{j}$-DTR of $K$ returns a ``Yes'' for $q$. In this case, the algorithm declares $\mathfrak{D}(q)=\de_{j}$, and the affirmative answer to the membership query says that $D_K(q)\geq (1-\eps)\,\de_j$. Due to the minimality of $j$, either $j=1$, or if not, 
        $K_{\de_{j-1}}$ is empty, or the $\eps$-AMQ for the $\delta_{j-1}$-DTR of $K$ returns a ``No'' for $q$.  In each of these cases, we have that $D_K(q)\leq (1-\eps)^{-1}\de_j$. Thus,
        \[(1-\eps)\,D_K(q)\leq \mathfrak{D}(q)\leq (1-\eps)^{-1}D_K(q).\]
    \end{enumerate}
Thus, the algorithm correctly answers the $\eps$-ADQ for $K$. The worst-case query time occurs when, at each iteration, the {algorithm} reports ``No''. In this case, the query time is
\[O\left(\sum_{k=1}^{{\lceil\log_2m\rceil}}\frac{1}{\eps\,\de_{j_k}}\ln\frac{1}{\de_{j_k}}\right),\]
where $j_k=\left\lceil\left(1-\dfrac{1}{2^k}\right)m\right\rceil$ for each $1\leq k\leq {\lceil\log_2m\rceil.}$
Using that $l\leq\log_{\,1-\eps}2\eps,$ a simple calculation shows that the query time is
\[O\left(\frac{1}{\eps^2}\left(\ln\frac{1}{\eps}\right)^2\right),\]
thus establishing Theorem~\ref{T:ADQ}, modulo Theorem~\ref{T:AMQ}. 

\section{An algorithm for answering the $\boldsymbol{\eps}$-AMQ for the $\boldsymbol{\de}$-DTR of $K$}\label{SS:DOA} Since depth-trimmed regions are affine invariant, we henceforth assume that $K$ is in canonical form, with $\de\in (0, 1/2]$ sufficiently small such that $K_{\de}\neq \emptyset,$ and $\eps\in (0,1).$ Recall that $E^{\lambda}(x)$ denotes the $\lambda$-Macbeath ellipsoid at $x.$ 

Given  constants $\lambda_p, \lambda_c>0,$ and a point $x\in K,$ define $E_p(x)=E^{\lambda_p}(x)$ and $E_c(x)=E^{\lambda_c}(x).$ A {\em$(\lambda_p, \lambda_c)$-Delone set}, or simply, a {\em Delone set} in $K_{\de}$ is a set of points $X=X(\de, \lambda_p, \lambda_c)\subset K_{\de}$ such that the following hold. 
\begin{enumerate}
\item[(a)] $X$ is {\em $\lambda_p$-packing}, i.e.,  $E_p(x)\cap E_p(y)=\emptyset$ for distinct points $x,y \in X,$ and for any $w\in K_{\de},$ there exists $z\in X$ such that $E_p(w)\cap E_p(z)\neq \emptyset.$
\item[(b)] $X$ is {\em $\lambda_c$-covering}, i.e., 
\[K_{\de}\subset \bigcup_{x\in X} E_c(x).\]
\end{enumerate}
The collections $\{E_p(x):x\in X\}$ and $\{E_c(x):x\in X\}$ will be referred to as the set of packing ellipsoids and covering ellipsoids, respectively, of $X$.

We rely on the following existence result. 

\begin{lemma}
\label{L:existenceDS}  Let $\lambda_p, \lambda_c\in (0,1)$ be such that
\be\label{E:lampc}
\frac{3+\lambda_p}{1-\lambda_p}\sqrt{d}\lambda_p<\lambda_c\leq\frac{e^{\eps\,\de\,|K|}-1}{e^{\eps\,\de\,|K|}+1}.
\ee
There exists a $(\lambda_p,\lambda_c)$-Delone set $X$ in $K_\de$ such that 
    \be\label{E:impinc}
      K_\delta\subset  \bigcup_{x\in X}E_c(x)\subset K_{(1-\eps)\,\de}. 
    \ee
In particular, there exists a $(\lambda_p,\lambda_c)$-Delone set $X_0$ in $K_\de$ corresponding to 
\bes 
\lambda_p=\dfrac{\eps\,\de\,|K|}{2(d^4+1)} \quad \text{and}\quad  \lambda_c=\dfrac{\eps\,\de\,|K|}{4},
\ees
that satisfies \eqref{E:impinc}.
\end{lemma}

Assuming this result for the moment, we describe our data structure and algorithm.

\subsection{The data structure and the algorithm} Let $X_0$ be as granted by Lemma~\ref{L:existenceDS}. The data structure consists of 
\begin{itemize}
 \item [(a)] the intersection graph, $G$, of the covering ellipsoids $\{E_c(x):x\in X_0\}$, whose vertex set is $X_0,$ and two distinct vertices $x,y\in X_0$ are adjacent in $G$ if and only if $E_c(x)\cap E_c(y)\neq \emptyset;$
    \item [(b)] for each vertex $x\in X_0,$ the coordinates of $x$ and the equation of its associated covering ellipsoid $E_c(x);$
    \item [(c)] the coordinates of the ``root''---an $x_{\delta}$ granted by Corollary~\ref{C:bound}.
\end{itemize} 

\begin{figure}[h]
\begin{overpic}[grid=false,tics=10,scale=0.75]{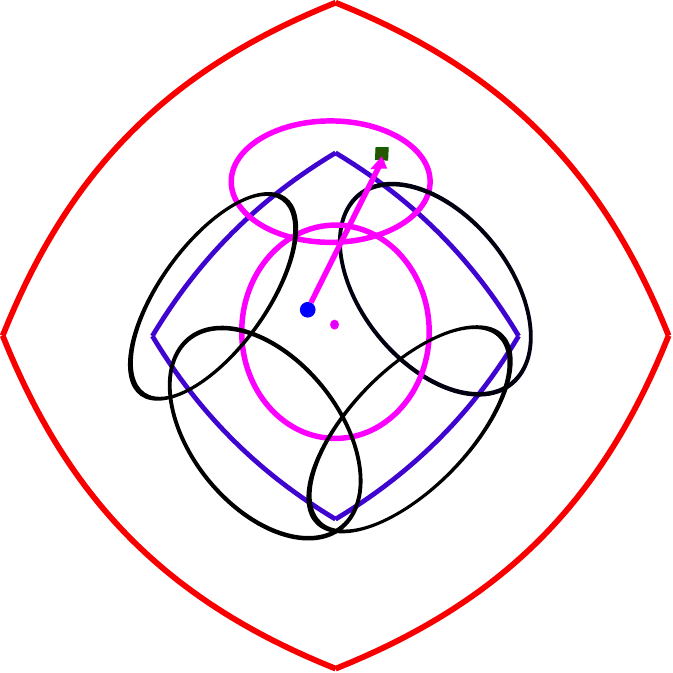}
    \put(41,50){$x_{\de}$}
    \put(50,48){$x_0$}
    \put(52,78){$q$}
    \put(28,54){$K_{\de}$}
    \put(42,91){$K_{(1-\eps)\,\de}$}
    \end{overpic}
    \caption{Search procedure for answering the $\eps$-AMQ for the $\de$-DTR of $K.$}
\label{fig:Algo}
\end{figure}

Let $x_0\in X_0$ be such that $x_{\de}\in E_c(x_0).$ To answer the $\eps$-AMQ for the $\de$-DTR of $K$ for a point $q\in\mathbb{R}^d,$ we walk along the ray starting at $x_{\de}$ and passing through $q$; see Figure~\ref{fig:Algo}. We start by setting $x\leftarrow x_0.$ If $q\in E_c(x),$ by \eqref{E:impinc}, we know that $q\in K_{(1-\eps)\,\de}$, and it is safe to report ``Yes.'' Otherwise, we inspect the vertices adjacent to $x$ in $G.$ If it is impossible to travel along the ray by passing through one of the covering ellipsoids associated with these vertices, then we stop and report ``No.'' Else, we determine the vertex $x_1$ adjacent to $x$ in $G$ for which the ray can travel the farthest through $E_c(x_1).$ We set $x\leftarrow x_1$ and repeat the above process. We set the maximum number of iterations (equivalently, the maximum number of ellipsoids visited during the search) to $N_0=\dfrac{d^{10d}}{\eps\,\de} \ln\dfrac{1}{\de}.$ If the search procedure does not yield a result after $N_0$ iterations, we terminate the search and report ``No.''  By \eqref{E:impinc}, this procedure reports a ``No'' if $q\notin K_{(1-\eps)\,\de}$.
In the next subsection, we show that it reports a ``Yes'' if $q\in K_{\de},$ thereby correctly answering the $\eps$-AMQ for the $\de$-DTR of $K.$

We now prove the existence result stated in Lemma~\ref{L:existenceDS}.

\noindent {\em Proof of Lemma~\ref{L:existenceDS}.} 
    Let $X$ be any $\lambda_p$-packing in $K_{\de}.$ We claim that $X$ is the desired Delone set. To see that it is $\lambda_c$-covering, take any point $y\in K_{\de}.$ Since $X$ is $\lambda_p$-packing, there exists $x\in X$ such that
\[E_p(x)\cap E_p(y)\neq \emptyset.\]
Applying Lemma~\ref{L:ellips} with $\alpha=1,$ $\lambda=\lambda_p$, and $\beta=\dfrac{3+\lambda_p}{1-\lambda_p}\sqrt{d},$ we have that 
\bes 
y\in E_p(y)\subset E^{\beta \lambda_p}(x)\subset E_c(x),
\ees
where the last containment follows from \eqref{E:lampc}.  Thus, $X$ is a $(\lambda_p, \lambda_c)$-Delone set in $K_{\de}.$ Let $z\in X$ be arbitrary. By \eqref{E:lampc} and Lemma~\ref{L:HilbMac}, $E_c(z)\subset B_K(z, \eps\,\de\,|K|/2).$ We claim that
\[B_K\left(z, \frac{\eps\,\de\,|K|}{2}\right)\subset K_{(1-\eps)\,\de}.\]
Then \eqref{E:impinc} follows. Let $w\in B_K(z, \eps\,\de\,|K|/2)$ be arbitrary. Note that
\[\frac{\eps\,\de\,|K|}{2}> \frac{1}{2}\ln\left(\frac{\|w-a\|\|z-b\|}{\|z-a\|\|w-b\|}\right)=\frac{1}{2}\ln\left(\left(1+\frac{\|z-w\|}{\|z-a\|}\right)\left(1+\frac{\|z-w\|}{\|w-b\|}\right)\right),\]
where $a, b\in bK$ denote the points where the line joining $z$ and $w$ meets $bK$ so that $a, z, w$ and $b$ are in consecutive order. Since $K$ is in canonical form, it follows that
\[\frac{\eps\,\de\,|K|}{2}>\ln(1+\|z-w\|).\]
Thus, $\|z-w\|<\eps\,\de\,|K|.$ Let $H_+$ be any arbitrary closed halfspace whose boundary hyperplane $H$ contains $w.$ Let $H_+'$ be the translate of $H_+$ whose boundary hyperplane $H'$ contains $z.$ Since $z
\in K_{\de},$ we have that
\[|K\cap H_+'|\geq \de\,|K|.\]
There are two possible cases. First, suppose that  $|K\cap H_+|\geq |K\cap H_+'|.$  Then, it immediately follows that $|K\cap H_+| \geq \de\,|K|.$  Next, suppose that $|K\cap H_+|<|K\cap H_+'|.$ Since $K$ is in canonical form and $\|z-w\|<\eps\,\de\,|K|$, we have that
\[ |K\cap H_+'|<|K\cap H_+| + \eps\,\de\,|K|.\]
Thus, in both cases, \[|K\cap H_+|> (1-\eps)\,\de\,|K|.\]
Since $H_+$ was arbitrarily chosen, it follows that $w\in K_{(1-\eps)\,\de}.$ Thus,
\[B_K\left(z, \frac{\eps\,\de\,|K|}{2}\right)\subset K_{(1-\eps)\,\de},\]
and  \eqref{E:impinc} holds. Finally, note that if $\lambda_p=\dfrac{\eps\,\de\,|K|}{2(d^4+1)}$ and $\lambda_c=\dfrac{\eps\,\de\,|K|}{4},$ 
\bes
    \frac{3+\lambda_p}{1-\lambda_p}\sqrt{d}\lambda_p
    \leq \left(\frac{3d^4+4}{d^4+1}\right)\frac{\eps\,\de\,|K|}{2d^3\sqrt{d}}< \frac{\eps\,\de\,|K|}{4}=\lambda_c\leq\frac{e^{\eps\,\de\,|K|}-1}{e^{\eps\,\de\,|K|}+1}.
\ees
    Thus, $\lambda_p, \lambda_c\in (0,1)$ satisfy \eqref{E:lampc} and the final claim in Lemma~\ref{L:existenceDS} follows.
\qed

\subsection{Analysis: correctness}
We show that the search procedure reports a ``Yes'' if $q\in K_{\de},$ thereby correctly answering the $\eps$-AMQ for the $\de$-DTR of $K.$

 \begin{lemma} \label{L:QT} Given a query point $q\in K_{\de},$ the search procedure will report ``Yes.''
 \end{lemma}
 \begin{proof}Since the $E_c$ ellipsoids cover $K_{\de}$, it is enough to show that the search procedure for $q$ runs for at most $N_0$ iterations.
 If $q\in E_c(x_0),$ one iteration suffices. If not, let $x_1,..., x_{k}$ denote the vertices of $G$ whose associated covering ellipsoids are visited during the search after exiting $E_c(x_0)$. For $1\leq i\leq k,$ let $y_i$ denote the point where the search ray exits the ellipsoid $E_c(x_{i-1}).$ Since affine segments are geodesics in a Hilbert geometry,
\be\label{E:QT}
d_K(y_1,y_k)=\sum_{i=1}^{k-1}d_K(y_i,y_{i+1}).
\ee
Fix $1\leq i\leq k-1$. Using that $\ln(1+x)\geq x/2$ for $x\in (0,1)$, Lemma~\ref{L:HilbMac}, and the choice of $\lambda_p,$
\[B_K\left(y_i, \frac{\eps\,\de\,|K|}{8d^5}\right)\subset E_p(y_i).\]
Since $q\in K_{\de}$, by convexity,
it follows that $y_i\in K_{\de}$.  Following
the proof of Lemma~\ref{L:existenceDS}, there exists $x\in X_0$ such that $E_p(y_i)\subset E_c(x).$ Thus, $x$ is adjacent to $x_{i-1}$ in $G$ and we have that
\[d_K(y_i, y_{i+1})\geq \frac{\eps\,\de\,|K|}{8d^5}, \quad \text{for all }1\leq i\leq k-1.\]
By \eqref{E:QT},
\[k+1\leq 2+\frac{8d^5}{\eps\,\de\,|K|}d_K(y_1,y_k)\leq 2+\frac{8d^5}{\eps\,\de\,|K|}d_K(x_{\de}, q),\]
where $x_{\de}$ is the root. By Corollary~\ref{C:bound} and using that $K$ is in canonical form,
\[k+1\leq 2+\frac{8d^5}{\eps\,\de\,|K|}\left(\frac{1}{2}\ln\frac{1}{\de}+24d^3\right)\leq \frac{d^{10d}}{\eps\,\de}\ln \frac{1}{\de}=N_0.\]
Since the number of iterations is $k+1,$ this completes the proof.
\end{proof}

\subsection{Analysis: storage space and query time}
We first show that the degree of each vertex of $G$ is bounded above by a dimensional constant.

\begin{lemma}\label{L:degree} The maximum degree of a vertex of $G$ is $O(1).$ 
\end{lemma}
\begin{proof}
Consider any vertex $x$ of $G,$ and let $A(x)$ denote the set of vertices that are adjacent to $x$ in $G,$ i.e., $A(x)=\left\{y\in X_0: E_c(x)\cap E_c(y)\neq \emptyset\right\}.$ In Lemma~\ref{L:ellips}, take $\alpha=\dfrac{\lambda_p}{\lambda_c}$, $\lambda=\lambda_c,$ and $\beta=\dfrac{2+\alpha(1+\lambda_c)}{1-\lambda_c}\sqrt d$ to obtain that
\bes
E_p(y)=E^{\alpha \lambda_c}(y)\subset E^{\beta \lambda_c}(x), \quad \text{for all }y\in A(x).
\ees
Interchanging the roles of $x$ and $y,$ $E_p(x)\subset E^{\beta\lambda_c}(y)$, for all $y\in A(x)$. Next, note that
\bes
\left|E_p(y)\right|
=\left(\frac{\lambda_p}{\beta\lambda_c}\right)^d\left|E^{\beta\lambda_c}(y)\right|
\geq \left(\frac{\lambda_p}{\beta\lambda_c}\right)^d \left|E_p(x)\right|
=\left(\frac{\lambda_p}{\beta\lambda_c}\right)^{2d}\left|E^{\beta\lambda_c}(x)\right|,
\ees
for all $y\in A(x)$. Thus, since $E^{\beta\lambda_c}(x)$ contains the disjoint union of ellipsoids $E_p(y)$, $y\in A(x)$, 
    \bes
        \# A(x)\leq \left(\frac{\beta\lambda_c}{\lambda_p}\right)^{2d}=O(1),
    \ees
where we used that 
$\lambda_p=\dfrac{\eps\,\de\,|K|}{2(d^4+1)}$, $\lambda_c=\dfrac{\eps\,\de\,|K|}{4},$ and $\beta\leq 8\sqrt{d}$.
\end{proof}

The query time is bounded above by the product of the maximum number of ellipsoids visited during the search and the maximum degree of a vertex of $G.$
Due to the bound imposed in the algorithm on the number of iterations (ellipsoid visits) and Lemma~\ref{L:degree}, we conclude that  the query time is $O\left(\dfrac{1}{\eps\,\de} \ln \dfrac{1}{\de}\right)$.

Furthermore, by Lemma~\ref{L:degree}, each vertex stores only a constant amount of additional information about the adjacent vertices (their coordinates and the equation of the associated covering ellipsoid). Thus, it suffices to bound $\#X_0$ to obtain the upper bound on the storage space.

\begin{lemma}\label{L:FinalDelone}$\# X_0 = O\left(\dfrac{1}{\eps^d\,\de^{\frac{3d-1}{2}}}\right).$
\end{lemma}

\begin{proof} We claim that
    \bes\#X_0=O\left(\frac{1}{\lambda_p^d}\vol_K B_K\left({x_{\de}}, \frac{1}{2}\ln\frac{1}{\de}+\nu_d\right)\right),
    \ees
    where $x_\de$ is the root and $\nu_d>0$ is a dimensional constant. The result then follows from Lemma~\ref{L:VolBound} and the choice of $\lambda_p.$ To prove the claim, note that 
by Lemma~\ref{L:HilbMac}, for each $z\in X_0,$
\bes
B_K\left(z,\frac{1}{2}\ln\left(1+\frac{\lambda_p}{\sqrt{d}}\right)\right)\subset E_p(z).
\ees
By Lemma~\ref{L:VolBound} and using that $\ln(1+x)\geq x/2$ for $x\in(0,1)$, \be\label{E:ELB} 
\vol_K E_p(z)=\Omega(\lambda_p^d).
\ee
By Corollary~\ref{C:bound} and Lemma~\ref{L:HilbMac},
\bes
E_p(z)\subset
    B_K\left({x_{\de}}, \frac{1}{2}\ln\frac{1}{\de}+24d^3+\frac{1}{2}\ln\frac{1+\lambda_p}{1-\lambda_p}\right),\ees
where $x_{\de}$ is the root.
Since $\lambda_p< 1/2,$ we have that
\bes
E_p(z)\subset
    B_K\left({x_{\de}}, \frac{1}{2}\ln\frac{1}{\de}+\nu_d\right),\ees 
 where $\nu_d>0$ is a dimensional constant. Since $X_0$ is $\lambda_p$-packing, the claim follows from \eqref{E:ELB}.
    \end{proof}
    
\begin{remark}
Optimizing the choice of $\lambda_p,\lambda_c\in(0,1)$ satisfying \eqref{E:lampc} does not improve the bounds on the storage space and query time (in $\eps$ {and $\de$}).
\end{remark}


\section{Exact depth queries for planar polygons}\label{S:EDQ}

        \noindent We show that any open planar convex polygon $K,$ given as the intersection of $n$ halfplanes, can be preprocessed in $O(n\log n)$ time into a data structure of size $O(n)$ that supports an algorithm which answers the exact depth query for $K$ in $O(n^3)$ time, for any query point $q \in \mathbb{R}^2$.
        
Let $K$ be an open planar convex polygon. We will rely on the fact that, for any $q\in K,$
\begin{itemize}
    \item [(a)] $D_K(q)\,|K|$ is the minimum possible area of any cap of $K$ cut off by a chord of $K$ passing through $q;$
    \item [(b)] there is a (possibly non-unique) chord $\ell$ of $K$ that cuts off a cap of $K$ of area $D_K(q)\,|K|$ and has $q$ as its midpoint.
\end{itemize}
This follows from the definition of halfspace depth and Lemma~\ref{L:FB}(b). The core idea of our algorithm is to consider each pair of distinct edges of  $K$, and identify all chords of $K$ that pass through the chosen pair and have $q$ as their midpoint. This is done by solving a constrained system of linear equations, which may have no solution, a unique solution, or multiple solutions. If no pair of edges admits a solution, then the algorithm reports $D_K(q) = 0.$ Otherwise, for each pair of edges where a solution exists, we select one such solution (we will later show that choosing one suffices, even when multiple solutions exist). Each solution corresponds to a chord of $K$, which partitions $K$ into two regions whose areas are computed. This yields a finite list of at most $2\displaystyle\binom{n}{2}$ areas associated to $q$, and the algorithm reports $D_K(q)$ as the least of these areas, normalized by $|K|$. The correctness of the algorithm for any $q\in \R^2$ follows from $(a)$ and $(b)$; we elaborate on the procedure below.

Let $K$ be an open planar convex polygon given as the intersection of $n$ halfplanes. Without loss of generality, assume that none of the boundary lines of these halfplanes is parallel to the $y$-axis. The data structure stores the value of $|K|$ and the edges $e_1,..., e_n$ of $K,$ listed in increasing order of the angle that their outer normal vectors make with the positive $x$-axis. For each $j \in \{1,..., n\}$, let $m_j, c_j, a_j, b_j \in \mathbb{R}$ be such that
\[e_j = \left\{ (x, y) \in \mathbb{R}^2 : y = m_j x + c_j,\ a_j \leq x \leq b_j \right\}.\]
The data structure requires $O(n)$ storage space and can be constructed in $O(n \log n)$ time: this involves sorting the boundary lines to determine their cyclic order, determining each edge from the intersection of adjacent lines, and computing $|K|$ using the shoelace formula.

To describe the algorithm, let
$q=(x_0,y_0)\in\R^2$ be a query point. We repeat the following procedure for each pair $(j,k)$ with $1 \leq j < k \leq n.$ Solve the following system of linear equations in the unknowns $u, v, w,$ and $z$:
    \bea\label{E:SLE}
        \begin{pmatrix}
            1 & 0 & 1 & 0 \\
            0 & 1 & 0 & 1\\
            -m_j & 1 & 0 & 0\\
            0 & 0 & -m_k & 1
        \end{pmatrix}
        \begin{pmatrix}
        u\\ v\\ w\\ z
        \end{pmatrix}=
        \begin{pmatrix}
        2x_0\\ 2y_0\\ c_j\\ c_k
        \end{pmatrix}
    \eea
    subject to the constraints:
    \beas\label{E:const}
        &a_j\leq u\leq b_j,&\\
        &a_k\leq w\leq b_k.&
    \eeas
Note that $(u_0,v_0,w_0,z_0)$ is a solution of the above constrained system if and only if $(u_0,v_0)\in e_j$, $(w_0,z_0)\in e_k$, and $q$ is the midpoint of the line segment joining $(u_0,v_0)$ and $(w_0,z_0)$. Now, there are three possible scenarios.
\begin{itemize}
    \item [(i)] If the constrained system admits no solution, discard the pair $(j, k)$.
    \item [(ii)] If the constrained system admits a unique solution $(u_0, v_0, w_0, z_0)$, compute the equation of the line $\ell$ passing through $(u_0, v_0)$ and $(w_0, z_0)$. Then compute the area of the polygonal region in $K$ bounded by $e_j,e_{j+1},..., e_k,$ and $\ell$, as well as the area of the complementary polygonal region in $K$. Set $v_{j,k}$ to be the minimum of these two areas.
    \item [(iii)] If the constrained system admits multiple solutions, choose any one and label it as $(u_0, v_0, w_0, z_0)$, and proceed as in (ii) to compute $v_{j,k}$. We will later show that all choices of solution yield the same value for $v_{j,k}$.
\end{itemize}
Once the above procedure has been carried out for all pairs $(j, k)$ with $1 \leq j < k \leq n$, if all such pairs are discarded, the algorithm reports $D_K(q) = 0$.  Otherwise, it reports \bes
    D_K(q) = \frac{1}{|K|} \min\{v_{j,k} : 1 \leq j < k \leq n\}.
\ees
The correctness of the algorithm follows from facts (a) and (b) stated at the beginning of the section. To bound the running time, observe that for each pair $(j,k)$, we first solve the constrained system, which takes $O(1)$ time. In the worst case, determining the line $\ell$, computing its intersections with $e_j$ and $e_k$, and performing two area computations using the shoelace formula to obtain $v_{j,k}$, all together take $O(n)$ time.
 Thus, the worst-case query time, including area computations for each pair $(j,k)$ and the final evaluation of $D_K(q)$, is bounded above by
\[
O\left(\binom{n}{2} \cdot n\right) = O(n^3).
\]

\begin{figure}[h]
\begin{overpic}[grid=false,tics=10,scale=0.5]{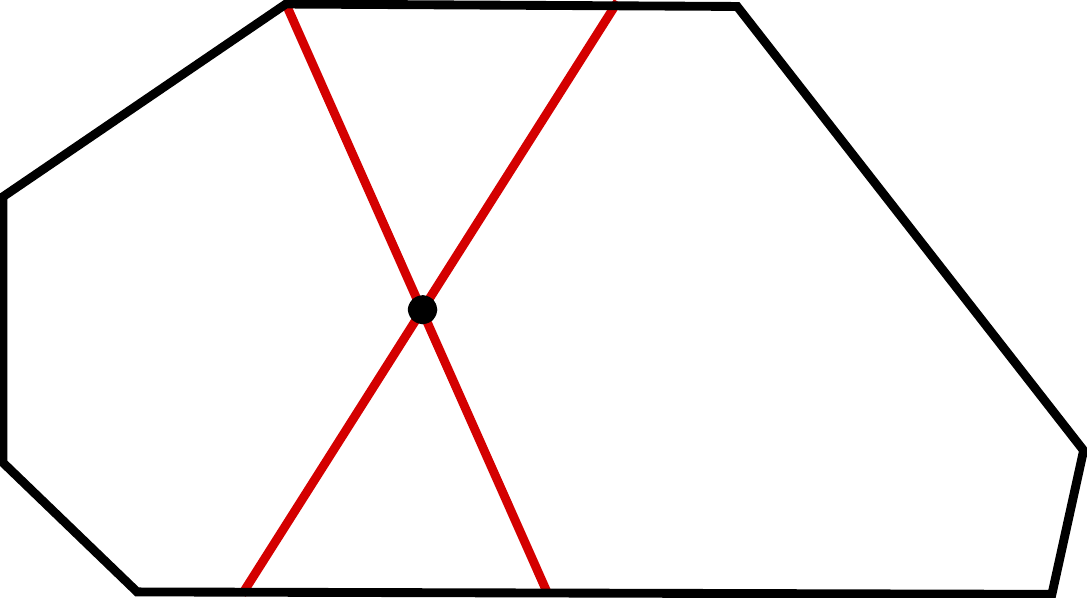}
\put(8,-6){$A$}
\put(96,-5.5){$B$}
\put(102,12){$C$}
\put(68,56.7){$D$}
\put(23,57){$E$}
\put(-6,37){$F$}
\put(-6,10){$G$}
\put(42,25){$q$}
\put(19,-6){$P$}
\put(51,-6){$Q$}
\put(56,57){$R$}
\end{overpic}
\bigskip 
\caption{Two solutions of the constrained system.}\label{F:poly}
\end{figure}

 It remains to show that any choice of solution in (iii) suffices. To prove this, observe that the constrained system admits multiple solutions if and only if the unconstrained system \eqref{E:SLE} admits infinitely many solutions. This, in turn, happens if and only if the $4 \times 4$ matrix on the left-hand side has rank less than $4$---specifically, rank $3$ in this case. Upon row reduction, we find that this occurs if and only if $m_j = m_k$, i.e., when $e_j$ and $e_k$ are parallel. In this case, any two chords of $K$ connecting $e_j$ and $e_k$ and having $q$ as their midpoint divide $K$ in the same proportion. This is evident from Figure~\ref{F:poly}, where $\triangle PqQ$ and $\triangle EqR$ are congruent by SAS congruence. Consequently,
\[
|QBCDE| = |EqR| + |QBCDRq| = |PqQ| + |QBCDRq| = |PBCDR|.
\]
Thus, in (iii), all choices of solution yield the same value for $v_{j,k}.$

\bibliography{AFMQ}{}
\bibliographystyle{plain}

\section*{Appendix: Deferred Proofs and Remarks}

\begin{proof}[Proof of Lemma~\ref{L:width}]
    This is trivial to see when $d=1$. Say $d\geq 2$. Let $w$ and $z$ denote the width and apex of the cap $C$, respectively. Let $H$ denote the hyperplane passing through the base of $C$ and let $H'$ denote the supporting hyperplane of $K$ at $z$, which is parallel to $H$. We claim that 
    \be\label{E:ray_width}
        \frac{\sqrt{3}\gamma}{d}w\leq \text{ray}(x)<
        \frac{1}{2\gamma}w.
    \ee
    Assuming \eqref{E:ray_width}, we complete the proof of \eqref{E:ray_vol}. Let $v$ denote the point where the ray $oz$ intersects the base of the cap $C$. There are two possibilities. Either $o\notin C$, in which case $v\neq o$, or $o$ lies in the base of $C$, in which case $v=o$. 

    In the former case, $\text{ray}(v)\geq w.$ By Lemma~\ref{L:width2}, we have that $\delta(v)\geq 2\gamma\,\text{ray}(v)\geq 2\gamma w.$ It follows that
    \be\label{E:volUB}
    |C|\geq \frac{1}{2}|\mathbb B(v,\de(v))|=2^{d-1}\omega_d\gamma^d w^d.\ee
    In the latter case, the lower bound on $|C|$ in \eqref{E:volUB} holds because $\mathbb B(o,\gamma)\subset K$ and $w<1/2$.
    
    Next, without loss of generality, let $H'$ be parallel to $\left\{(t_1,...,t_d)\in\mathbb{R}^d : 
    t_d=0\right\}$. Consider the $d$-dimensional cube $P=(-\frac{1}{2},\frac{1}{2})^d$. The cap $C$ is strictly contained in the slab $S$ of $P,$ of height $w,$ bounded by $H$ and $H
    '$. Thus,
    \be\label{E:volLB}
    |C|<|S|=w.
    \ee
    Combining \eqref{E:ray_width}, \eqref{E:volUB}, and \eqref{E:volLB}, we get the result. 

    We now prove \eqref{E:ray_width}. First, assume that $o\notin C$. Let $a$ and $b$ denote the points of intersection of the ray $ox$ with $b K$ and $H'$, respectively. Since $o\notin C$, the points $o, x, a$ and $b$ are in consecutive order. Thus,
    \bes
    \frac{\text{ray}(x)}{\Vert a\Vert}= \frac{\Vert x-a\Vert}{\Vert 
    a\Vert}\leq \frac{\Vert x-a\Vert +\Vert a-b\Vert}{\Vert a\Vert +\Vert a-b\Vert}=\frac{\Vert x-b\Vert}{\Vert b\Vert}=\frac{w}{\text{dist}(o, H')},\ees
    where the last equality follows from similarity of triangles. Since $\mathbb{B}(o, \gamma)\subset K\subset \mathbb{B}(o, 1/2),$ $\Vert a\Vert<1/2$ and $\text{dist}(o, H')\geq \gamma.$ Thus, we obtain the upper bound in \eqref{E:ray_width}.
    
\begin{figure}[h]
\centering
\begin{subfigure}{0.36\textwidth}
\includegraphics[width=\textwidth]{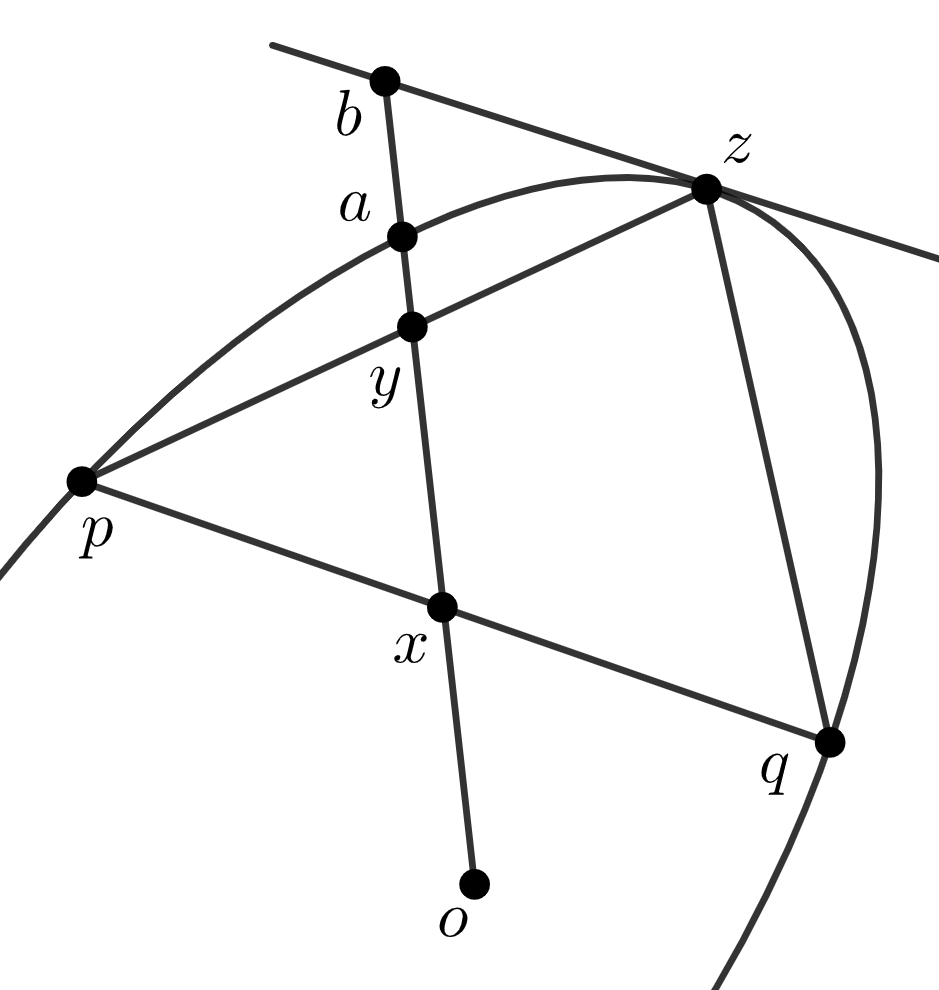}
    \caption{}
    \label{F:UB}
\end{subfigure}\hfill
\begin{subfigure}{0.55\textwidth}
\includegraphics[width=\textwidth]{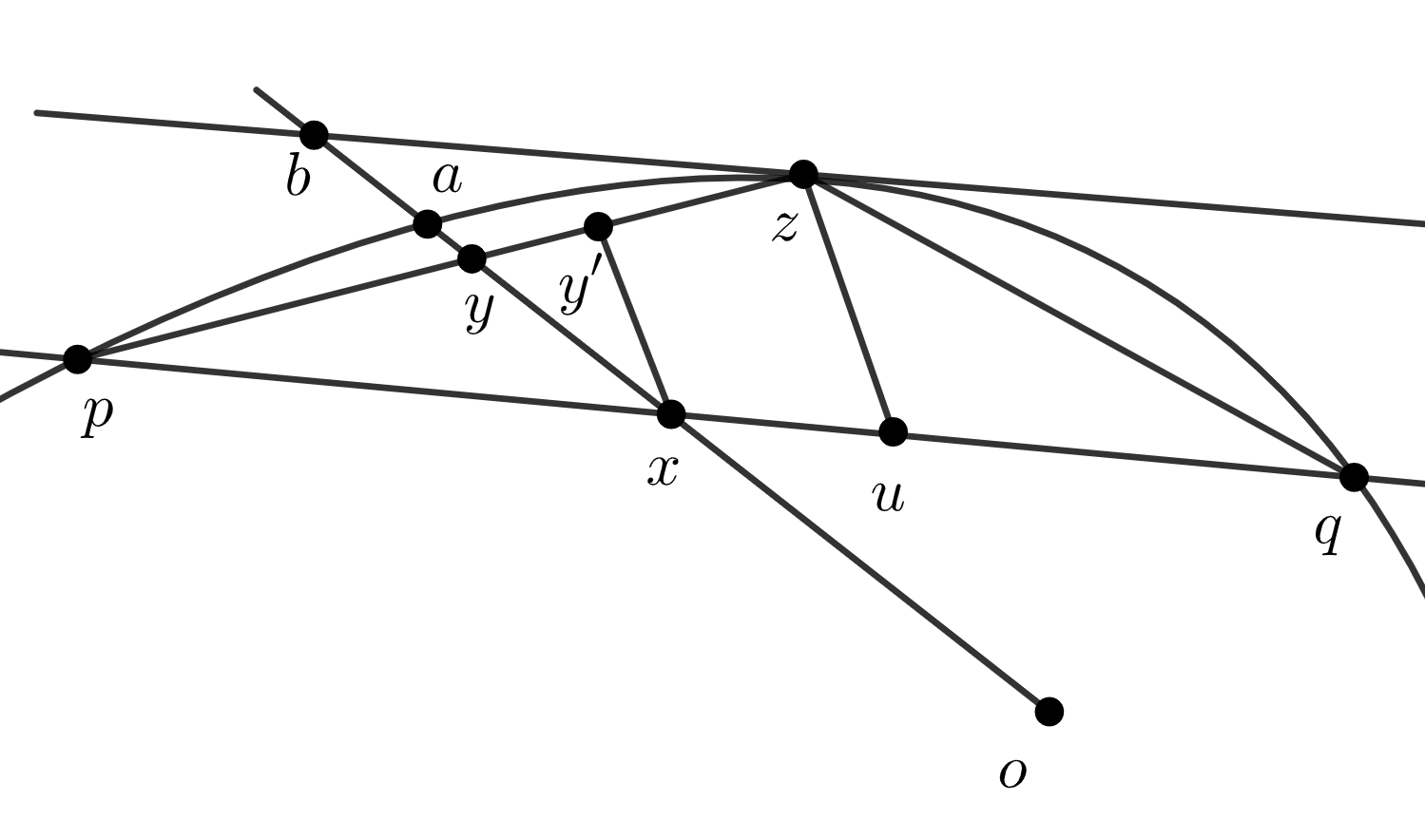}
    \caption{}
    \label{F:LB}
\end{subfigure}
        \caption{Proof of Lemma~\ref{L:width}}
\label{fig:figures}
\end{figure}
 
Next, consider a $2$-plane $\Pi$ containing the points $o,x$ and $z.$ The intersection of $\Pi$ with $K\cap H$ is a line segment which intersects $b K$ at, say $p$ and $q.$ Label them so that the ray $ox$ intersects $pz$ at a point $y$. Note that $x$ is the centroid of $K\cap H$ and lies on the line segment $pq$. By applying \cite[Theorem II]{Ha51} to the $(d-1)$-dimensional convex body $K\cap H$, we obtain that
\be\label{P1_2}\frac{\Vert p-x\Vert}{\Vert p-q\Vert}\geq \frac{1}{d}.\ee
Consider the case where $\angle pzq$ is acute; see Figure $1$(A).
Since $\mathbb{B}(o, \gamma) \subset K\subset \mathbb{B}(o, 1/2)$, $\angle pzq \geq 2\arcsin{\gamma}$. Since $\angle pzq\in (0, \pi/2]$, we have that $\sin \angle pzq \geq \sin(2\arcsin \gamma)$. Since $z$ is the apex of $C$, $\Vert z-q\Vert\geq w$. By \eqref{P1_2} and an application of the sine rule on $\Delta pzq$ and $\Delta pyx$,
\be\label{P1_3}\Vert y-x\Vert=\Vert z-q\Vert\frac{\Vert p-x\Vert}{\Vert p-q\Vert}\frac{\sin \angle pzq}{\sin \angle pyx}\geq \frac{\sin(2\arcsin \gamma)}{d}w=\frac{2\gamma\sqrt{1-\gamma^2}}{d}w\geq \frac{\sqrt{3}\gamma}{d}w.\ee
Now, consider the case where $\angle pzq$ is not acute; see Figure $1$(B).
There is a point $u$ on the segment $pq$ such that $\angle pzu=\pi/2.$ Drop a perpendicular from $x$ on the segment $pz$ to get a point $y'$ such that $\angle py'x=\pi/2.$ By \eqref{P1_2} and the similarity of $\Delta py'x$ and $\Delta pzu$,
\be\label{P1_4}\Vert y-x\Vert\geq \Vert y'-x\Vert=\frac{\Vert p-x\Vert}{\Vert p-u\Vert}\Vert z-u\Vert\geq \frac{\Vert p-x\Vert}{\Vert p-q\Vert}\Vert z-u\Vert\geq \frac{1}{d}w.\ee
By \eqref{P1_3} and \eqref{P1_4}, and the fact that $\text{ray}(x)\geq \Vert y-x\Vert$, we obtain the lower bound in \eqref{E:ray_width} when $o\notin C$.

{Now, assume that $o$ lies in the base of $C$. Then, $\gamma\leq  w< 1/2$. Let $r, s\in bK$ denote the points of intersection of the line joining $o$ and $x$ with $bK,$ so that $r,o,x,$ and $s$ are in consecutive order.  We apply \cite[Theorem II]{Ha51} to the $(d-1)$-dimensional convex body $K \cap H$, whose centroid is $x$, to obtain that
    \bes
        \text{ray}(x)=\|s-x\|\geq\frac{\|s-r\|}{d}=\frac{\|r\|+\|s\|}{d}\geq 
        \frac{2\gamma}{d}>\frac{\sqrt{3}\gamma}{d}w.
    \ees
On the other hand, 
    \bes
        \text{ray}(x)<\frac{1}{2}\leq \frac{1}{2\gamma}w.
    \ees
The proof is now complete.}
    \end{proof}

\noindent{\em Proof of Theorem~\ref{T:FB-HB}.} Once \eqref{E:main_inc} is established, the nonemptiness of the ball in the left-hand side follows immediately by requiring the radius to be positive. To prove \eqref{E:main_inc}, we first consider the case where $K_{\de}$ is a singleton, specifically $K_{\de}=\left\{o\right\}.$ Since any hyperplane through $o$ cuts off a cap of volume at least $\gamma^d\omega_d/2$ from $K$, we have that   \bes
    {\de\,|K|}\geq \frac{\gamma^d\omega_d}{2}\geq 2^{2d-1}\gamma^{3d}\omega_d,
    \ees
and thus, \eqref{E:main_inc} holds trivially.

Now, suppose $K_\de\neq\emptyset$ is not a singleton. Since $K_{\de}$ is strictly convex, it has nonempty interior. Let $o\neq z\in bK_\de$ be arbitrarily fixed. Then, by Lemma~\ref{L:FB}$(b)$, there is a supporting hyperplane of $K_\de$ at $z$ such that the cap of $K$, say $C$, cut off by this hyperplane, has volume $\de\,|K|$. Moreover, $z$ is the centroid of the base of $C.$ Since $o\in K_\de$, $o\neq z$, and $K_\de$ is strictly convex, $o\notin C.$ By Lemma~\ref{L:width}, %
\be\label{E:ray} \eta\,\de\,|K|< \text{ray}(z)< \zeta(\de\,|K|)^{1/d},
\ee%
where $\eta=\dfrac{\sqrt{3}\gamma}{d}$ and $\zeta=\dfrac{1}{4\gamma^2}\left(\dfrac{2}{\omega_d}\right)^{1/d}.$ We claim that \be\label{E:dilate}
        \left(1-\frac{\zeta}{\gamma}(\de\,|K|)^{\, 1/d}\right)K\subset K_\de\subset (1-2\eta\,\de\,|K|)K,
    \ee
where, note that 
\begin{itemize}
\item [(a)]
$2\eta\,\de\,|K|< 2\,\text{ray}(z)<1,$ because $K\subset \mathbb{B}(o, 1/2)$,
\smallskip
\item [(b)] when $o\in bK_\de$, $\dfrac{\zeta}{\gamma}(\de\,|K|)^{\, 1/d}\geq 1$ since the condition $\mathbb B(o,\gamma)\subset K$ yields $\de\,|K|\geq \dfrac{\gamma^d\omega_d}{2},$ and thus, the left-hand side containment in ~\eqref{E:dilate} holds trivially.
\end{itemize}
Assuming \eqref{E:dilate} for the moment, Lemma~\ref{L:DilatesHilb} gives that
\bes B_K\left(o, \frac{1}{2d}\ln \frac{1}{{\de\,|K|}}+p\right)\subset K_\de\subset B_K\left(o, \frac{1}{2}\ln\frac{1}{{\de\,|K|}}+q\right),
\ees
where $p=\dfrac{1}{2}\ln\dfrac{\gamma}{\zeta}=\dfrac{1}{2d}\ln(2^{2d-1}\gamma^{3d}\omega_d)$ and $q= \dfrac{1}{2}\ln\left(\dfrac{2\gamma+1}{4\eta\gamma}\right)= \dfrac{1}{2}\ln \dfrac{d(2\gamma+1)}{4\sqrt{3}\gamma^2},$ thus proving \eqref{E:main_inc}.
\begin{figure}[h]
\centering
\begin{subfigure}{0.5\textwidth}
\begin{overpic}[grid=false,tics=10,scale=0.17]{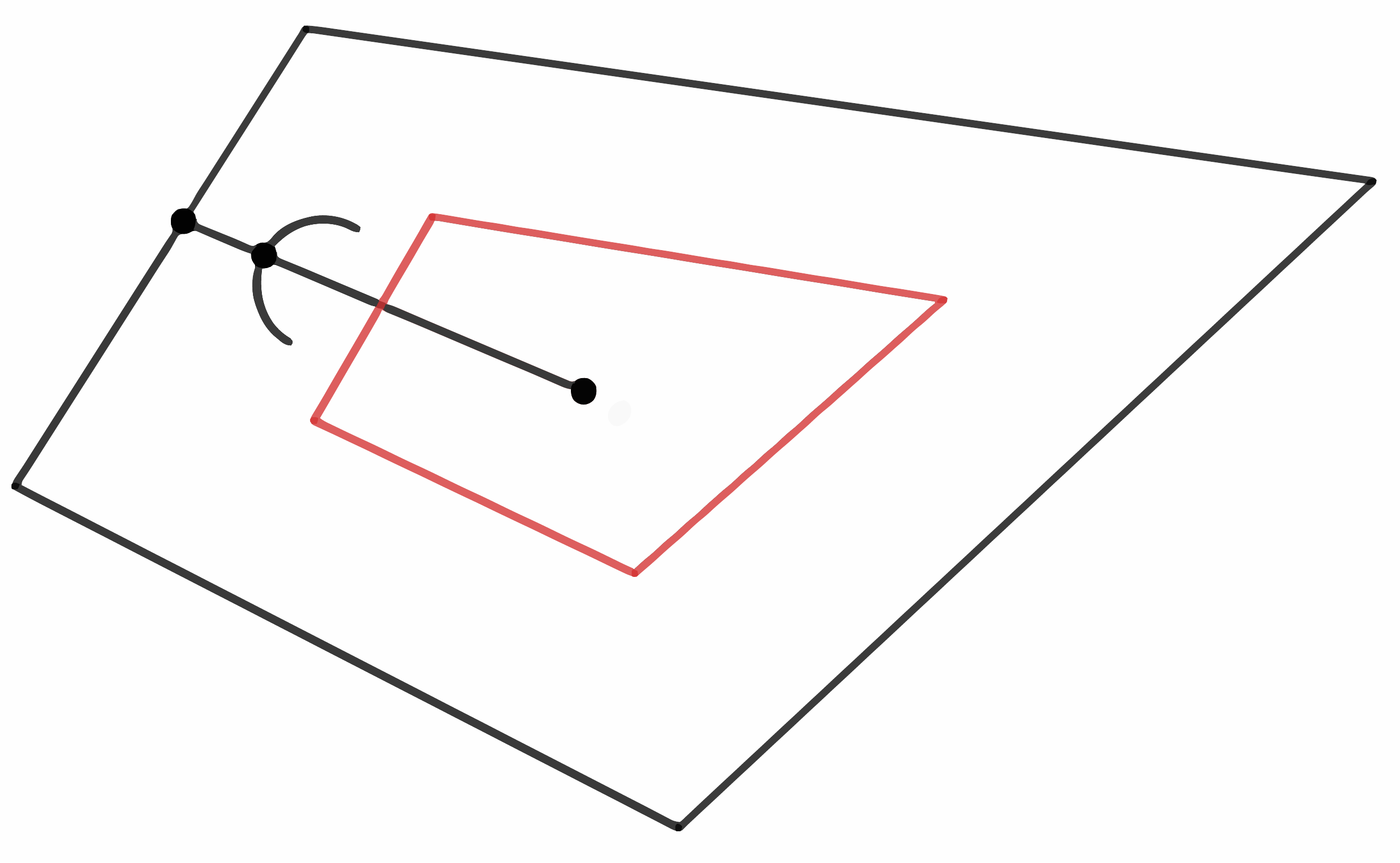}
\put(10,50){$y$}
\put(16,47){$y_{\de}$}
\put(12,32){$bK_{\de}$}
\put(42,30){$o$}
\put(60,44){$(1-t_{\de})K$}
\put(40,10){$K$}

\end{overpic}
    \caption*{(A)}
\end{subfigure}\hfill
\begin{subfigure}{0.5\textwidth}
\begin{overpic}[grid=false,tics=10,scale=0.17]{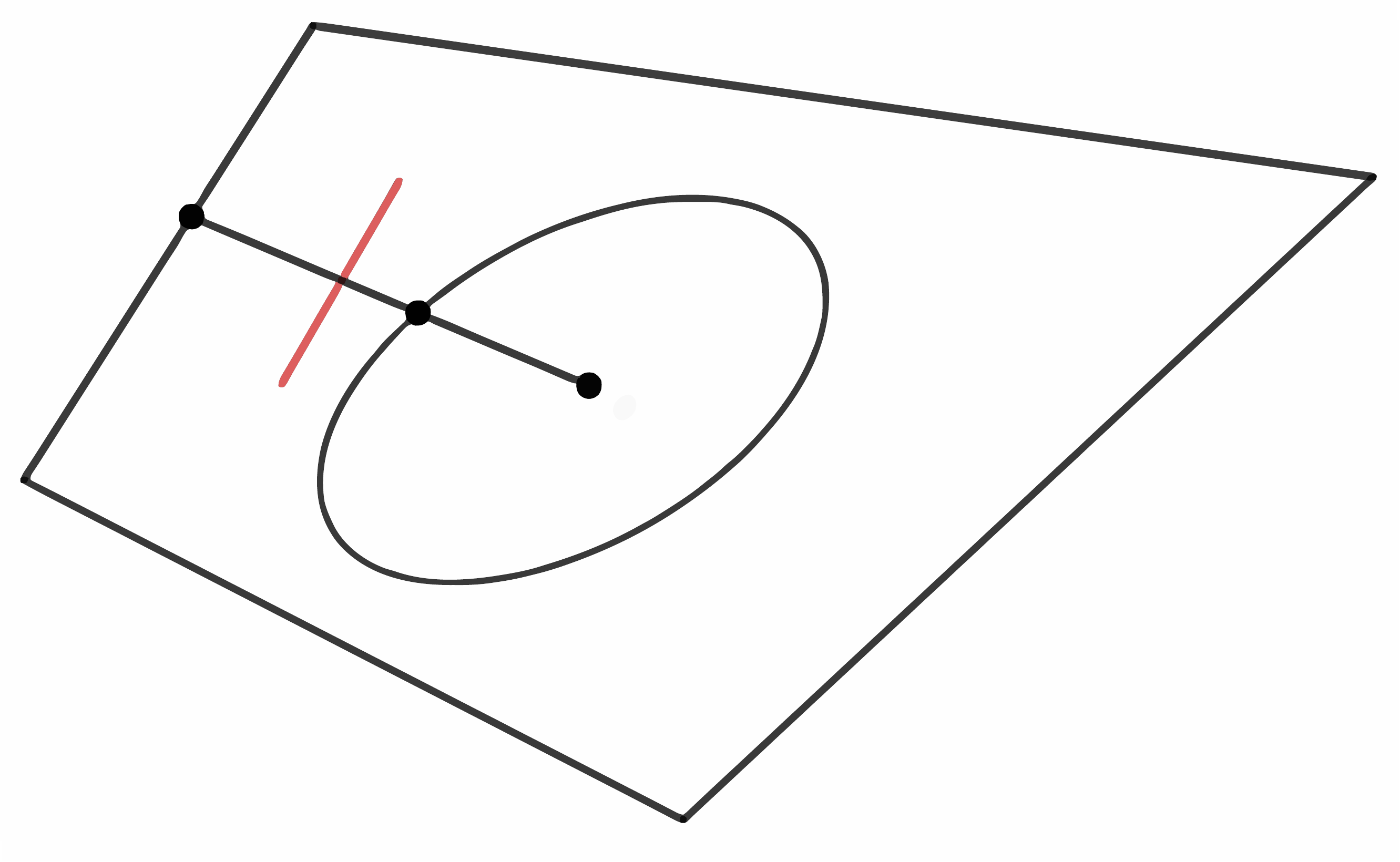}
\put(10,48){$z$}
\put(27,42){$z_{\de}$}
\put(42,30){$o$}
\put(40,11){$K$}
\put(60,44){$K_\de$}
\put(20,50){$(1-s_{\de})\,bK$}
\end{overpic}
    \caption*{(B)}
\end{subfigure}
        \caption{Proof of Theorem~\ref{T:FB-HB}}
\label{F:FB-HB}
\end{figure}

To establish the right-hand side containment in \eqref{E:dilate}, suppose that, for some $t_{\de}\in(0,1)$, we have $K_\de\not\subset (1-t_{\de})K.$ Then, there exists $y_{\de}\in  b  K_\de$ such that $y_{\de}\notin (1-t_{\de})K.$  Since $o\in (1-t_{\de})K$, $y_\de\neq o.$ Let $y\in b  K$ denote the point where the ray from $o$ to $y_{\de}$ intersects $ bK$; see Figure~\ref{F:FB-HB}(A). Then,
\[(1-t_{\de})\|y\|\leq \|y_{\de}\|=\|y\|-\text{ray}(y_{\de}).\]
Since $K\subset \mathbb B(o,1/2),$ by \eqref{E:ray},
\[1-t_{\de}< 1-2\eta\,\de\,|K|\implies t_{\de}>2\eta\,\de\,|K|.\]
This yields the right-hand side containment in \eqref{E:dilate}. 

Next, we establish the left-hand side containment in \eqref{E:dilate}. Suppose that, for some $s_{\de}\in(0,1)$, we have $(1-s_{\de})K\not\subset K_\de.$ Then, there exists $z\in b K$ such that $(1-s_{\de})z\notin K_\de.$ Let $z_{\de}\in  b  K_\de$ denote the point where the ray from $o$ to $z$ intersects $ b  K_\de$; see Figure~\ref{F:FB-HB}(B). It is possible that $z_\de=o$. Then,
\[(1-s_{\de})\|z\|>\|z_{\de}\|=\|z\|-\text{ray}(z_{\de}).\]
Since $\mathbb B(o,\gamma)\subset K,$ by \eqref{E:ray},
\[1-s_{\de}>1-\frac{\zeta}{\gamma}(\de\,|K|)^{1/d}\implies s_{\de}<\frac{\zeta}{\gamma}{\left(\de\,|K|\right)}^{1/d}.\]
This yields the left-hand side containment in \eqref{E:dilate}. The proof of \eqref{E:main_inc} is now complete. \qed

\noindent {\em On the sharpness of Theorem~\ref{T:FB-HB}.} The containments in \eqref{E:dilate} (and, therefore, in \eqref{E:main_inc}) are sharp in the 
order of $\de$. For instance, consider the $d$-dimensional cube $K=\left(-\frac{1}{2},\frac{1}{2}\right)^d$. Since $|K|=1$ and $K=-K$, its $\de$-DTR coincides with its $\de$-convex floating body, which coincides with its (Dupin) $\de$-floating body; see \cite{MeRe91} for the latter statement. Thus, for $\de\in (0, 1/2]$, the centroid of the base of every cap of $K$ of volume $\de$
is an element of $bK_\de$. The hyperplanes 
    \beas
        H_1&=&\left\{(x_1,...,x_d)\in\Rd:x_d=\frac{1}{2}-\de\right\},\\
        H_2&=&\left\{(x_1,...,x_d)\in \Rd: x_1+\cdots+x_d=\frac{d}{2}-(d!\de)^{1/d}\right\},
    \eeas
cut off caps of volume $\de$ from $K$, and have centroids $\left(0,...,0,\dfrac{1}{2}-\de\right)$ and $\left(x(\de),...,x(\de)\right)$ with $x(\de)=\dfrac{1}{2}-\dfrac{(d!)^{1/d}}{d}\de^{1/d}$, respectively. From this, it follows that if $(1-\phi_\de)K\subset K_\de\subset (1-\varphi_\de)K$, then $\varphi_\de$ cannot be larger than $2\de$ and $\phi_\de$ cannot be smaller than $\dfrac{2(d!)^{1/d}}{d}\de^{1/d}.$
\end{document}